\documentclass[preprint]{elsarticle}
\usepackage{lineno,hyperref}
\usepackage{amsthm,amsmath,amssymb,amsbsy,amsfonts}
\usepackage{lscape}
\usepackage{pdflscape}
\theoremstyle{plain}
\newtheorem{thm}{Theorem}[section]
\theoremstyle{definition}

\newtheorem{lemma}{Lemma}[section]

\newtheorem{example}{Example}[section]
\usepackage{mathptmx,latexsym,amsopn,amstext,amsxtra,euscript,amscd}
\usepackage{graphicx,fancybox,tikz}
\usepackage[left=3.50cm, right=3.5cm, bottom=3.30cm]{geometry}
\begin{document}

\begin{frontmatter}

\title{Self-dual cyclic codes over $M_2(\mathbb{Z}_4)$}

\author[addre]{Sanjit Bhowmick}
\ead{sanjitbhowmick392@gmail.com}
\address[addre]{
Department of Mathematics,
National Institute of Technology Durgapur,
Durgapur, INDIA}

\author[addre]{Satya Bagchi}
\ead{satya.bagchi@maths.nitdgp.ac.in}

\author[add]{Ramakrishna Bandi\corref{cor}}
\cortext[cor]{Corresponding author}
\ead{ramakrishna@iiitnr.edu.in}
\address[add]{Department of Mathematics,
Dr. SPM International Institute of Information Technology,
Naya Raipur, INDIA}

\begin{abstract}
In this paper, we study the codes over the matrix ring over $\mathbb{Z}_4$, which is perhaps the first time the ring structure $M_2(\mathbb{Z}_4)$ is considered as a code alphabet. This ring is isomorphic to $\mathbb{Z}_4[w]+U\mathbb{Z}_4[w]$, where $w$ is a root of the irreducible polynomial $x^2+x+1 \in \mathbb{Z}_2[x]$ and $U\equiv$ ${11}\choose{11}$. We first discuss the structure of the ring $M_2(\mathbb{Z}_4)$ and then focus on algebraic structure of cyclic codes and self-dual cyclic codes over $M_2(\mathbb{Z}_4)$. We obtain the  generators of the cyclic codes and their dual codes. Few examples are given at the end of the paper.
\end{abstract}
\end{frontmatter}


\section{Introduction}
Codes over finite rings are very old on the other hand their applications in digital communication is rather young. Codes over rings have got the attention of researchers only after Hammons et. al. \cite{Hammons1994} in which they have shown an interesting relation between popular non-linear codes and linear codes over integer residue rings modulo 4, via a map called Gray map. This attracted the researchers to focus on codes over rings and their applications.  As a result so many new ring structures have been considered as code alphabets. However most of the study was restricted to codes  over commutative rings    \cite{Hammons1994,Pless1997,Pless1996}. Recently, codes over a non-commutative ring, the matrix ring over $\mathbb{Z}_2$, i.e., $M_2(\mathbb{Z}_2)$ has been considered as a code alphabet to study space time codes \cite{Oggier2012}. The advantage of this type of matrix rings is that they are non-commutative, which allows a quotient ring to have wide  left and right ideals (cyclic codes). This is also true in the case of skew polynomial rings, however, the polynomial factorization is a big hurdle in the construction of  codes over skew polynomial rings, which is not the case with codes over matrix rings. A notable work on cyclic codes over non-commutative finite rings is \cite{Bachoc1997,Greferath1999,Wisbauer1991}.

In \cite{Greferath1999}, the authors have studied cyclic codes over $M_2(\mathbb{Z}_2)$ and obtained some optimal codes over the same. Inspired by this work  Luo and Uday \cite{Luo2017} have obtained the structure of cyclic codes over $M_2(\mathbb{Z}_2+u\mathbb{Z}_2)$ and found some optimal cyclic codes. Motivated by this, in this paper, we explored the construction of codes over  $M_2(\mathbb{Z}_4)$. The reason for choosing $\mathbb{Z}_4$ is that  $\mathbb{Z}_4$ is best suited for the construction of modular lattices and also the relation established by Hammons et. al. \cite{Hammons1994} between binary non-linear codes and linear codes over $\mathbb{Z}_4$.  The approach which is being used in this paper to cyclic codes over $M_2(\mathbb{Z}_4)$ is same as that of cyclic codes over  $M_2(\mathbb{Z}_2)$, however, it is not straight forward as it can be seen later in the paper.

The paper is organised as follows: In Section 2, we describe the structure of $M_2(\mathbb{Z}_4)$ and  show that $M_2(\mathbb{Z}_4)$ is isomorphic to $\mathbb{Z}_4[w]+U\mathbb{Z}_4[w]$, where $w$ is a root of the polynomial $x^2+x+1$ and $U\equiv$ ${11}\choose{11}$. 
We define  a Gray map on $\mathbb{Z}_4[w]+U\mathbb{Z}_4[w]$ to $\mathbb{F}^4_4$ which preserves the Lee weight in $\mathbb{Z}_4[w]+U\mathbb{Z}_4[w]$ and Hamming weight in $\mathbb{F}^4_4$. In section 3, we discuss the structure of cyclic codes and prove some results on the dimension of cyclic codes. In section 4, we obtain the structure of dual cyclic codes and  also self-dual codes. 


\section{Structure of $M_2(\mathbb{Z}_4)$}
Let us denote $\mathcal{R}=M_2(\mathbb{Z}_4)$. $\mathcal{R}$ is a non-commutative ring of matrices of order 2 over $\mathbb{Z}_4$. The set $\mathbb{Z}_4+X\mathbb{Z}_4+Y\mathbb{Z}_4+YX\mathbb{Z}_4$ forms a non-commutative finite ring with respect to component wise addition and multiplication defined in Table \ref{p1_table1}.
\begin{table}[ht]
\begin{center}
\begin{tabular}{|c|| c| c| c| c| }\hline $\mathcal{R}$
 $\cdot$  &  1   &  X  &  Y  &  YX \\ \hline \hline
1  &  1   &  X  &  0  &  0  \\ \hline
X  &  0   &  0  &  1  &  X  \\  \hline
Y  &  Y   &  YX &  0  &  0  \\ \hline
YX &  0   &  0  &  Y  &  YX \\ \hline
\end{tabular}
\caption{Multiplication rule of $\mathcal{R}$}\label{p1_table1}
\end{center}
\end{table} 
\vspace{-0.5cm}
\begin{lemma}\label{p1_lemma1}
The ring $M_2(\mathbb{Z}_4)$ is isomorphic to the ring $\mathbb{Z}_4+X\mathbb{Z}_4+Y\mathbb{Z}_4 + YX\mathbb{Z}_4$, i.e., $M_2(\mathbb{Z}_4) \equiv \mathbb{Z}_4+X\mathbb{Z}_4+Y\mathbb{Z}_4+YX\mathbb{Z}_4$.
\end{lemma}

\begin{proof}
We define a mapping $f:  M_2(\mathbb{Z}_4)  \longrightarrow  \mathbb{Z}_4+X\mathbb{Z}_4+Y\mathbb{Z}_4+YX\mathbb{Z}_4$ such that $f(A) = a+Xb+Yc+YXd$,
where $A=\left(\begin{array}{cc} a & b  \\ c & d  \end{array}\right) \in M_2(\mathbb{Z}_4)$.
 It is easy to see that $f(A+B)=f(A)+f(B)$. Now we show that $f(AB)=f(A)f(B)$.  Let $A=\left(\begin{array}{cc} a & b  \\ c & d  \end{array}\right)$  and $B=\left(\begin{array}{cc} a_1 & b_1  \\ c_1 & d_1  \end{array}\right)$. Then $AB=\left(\begin{array}{cc} aa_1+bc_1 & ab_1+bd_1  \\ ca_1+dc_1 & cb_1+dd_1  \end{array}\right)$. So $f(AB)=(aa_1+bc_1) +X(ab_1+bd_1) +Y(ca_1+dc_1) +YX(cb_1+dd_1)$.

Now,
\begin{equation*}
\begin{array}{rcl}
f(A)f(B) & = & (a +Xb +Yc +YXd)(a_1 +Xb_1 +Yc_1 +YXd_1)\\
                & = & (aa_1+bc_1) +X(ab_1+bd_1) +Y(ca_1+dc_1) +YX(cb_1+dd_1 ) \\
                & = & f(AB).
\end{array}
\end{equation*}
It is easy to see that $f$ is one-one and onto. 
Hence $M_2(\mathbb{Z}_4) \equiv \mathbb{Z}_4+X\mathbb{Z}_4+Y\mathbb{Z}_4+YX\mathbb{Z}_4$.
\end{proof}
We consider a subset  of $\mathbb{Z}_4 + X\mathbb{Z}_4 + Y\mathbb{Z}_4$ $+YX\mathbb{Z}_4$, namely $W = \lbrace 0$, $X+Y3+YX3$, $X2+Y2+YX2$, $X3+Y+YX$, $1+X+YX3$, $2+X2+YX2$, $3+X3+YX$, $1+YX$, $2+YX2$, $3+YX3$, $1+X2+Y2+YX3$, $3+X2+Y2+YX$, $1+X3+Y+YX2$, $2+X3+Y+YX3$, $2+X+Y3+YX$, $3+X+Y3+YX3 \rbrace$.

\begin{lemma}\label{p1_lemma2}
The subset $W$ forms a commutative ring with respect to component wise addition and multiplication defined on $\mathbb{Z}_4+X\mathbb{Z}_4+Y\mathbb{Z}_4+YX\mathbb{Z}_4$. 
\end{lemma}

\begin{proof}
One can easily verify that $W$ is an abelian group under component wise addition. The other criteria of a ring can be verified using the Table \ref{p1_table2}. For simplicity, we use the following notation:\\
$\begin{array}{cccc}
a_0=0, & a_{1}= X+Y3+YX3, & a_{2}= X2+Y2+YX2, & a_{3}= X3+Y+YX, \\
a_{4}=1+X+YX3, & a_{5}=2+X2+YX2, & a_{6}= 3+X3+YX, & a_{7}= 1+YX, \\
a_{8}=2+YX2, & a_{9}= 3+YX3, & a_{10}= 1+X2+Y2+YX3, & a_{11}=3+X2+Y2+YX, \\
a_{12}=1+X3+Y+YX2, & a_{13}=2+X3+Y+YX3, & a_{14}=2+X+Y3+YX, & a_{15}=3+X+Y3+YX3.  
\end{array}$

\end{proof}

\begin{center}
\begin{table}[ht]
\begin{tabular}{| c|| c| c| c| c| c| c| c| c| c| c| c| c| c| c| c| c | }\hline
$\cdot$ & 0 & $a_1$  & $a_2$ &$a_3$ & $a_4$ & $a_5$ &$a_6$ &$a_7$ &$a_8$ &$a_9$ &$a_{10}$ &$a_{11}$ &$a_{12}$ &$a_{13}$ & $a_{14}$ & $a_{15}$ \\\hline\hline  
 0  &  0  &  0 &  0 & 0  &  0  &  0 &  0 & 0  &  0  &  0 &  0 & 0  &  0  &  0 &  0 & 0\\  \hline
$a_1$ & 0 & $a_6$ &$a_5$ &$a_4$ &$a_9$ &$a_8$ &$a_7$ & $a_1$ &$a_2$ &$a_3$ &$a_{13}$ & $a_{14}$ &$a_{10}$ &$a_{12}$ & $a_{15}$&$a_{11}$ \\ \hline
  
$a_2$ &  0 & $a_5$ & 0 &$a_5$ &$a_8$ & 0 &$a_8$ &$a_2$ & 0 &$a_2$ &$a_2$ &$a_2$ & $a_{8}$ &$a_5$ &$a_5$ &$a_8$\\ \hline 
 
$a_3$ &  0 & $a_4$ &$a_5$ &$a_6$ &$a_7$ &$a_8$ &$a_9$ &$a_3$ &$a_2$ & $a_1$ & $a_{14}$ &$a_{13}$ &$a_{11}$ & $a_{15}$ &$a_{12}$ &$a_{10}$\\ \hline 
 
$a_4$ &  0 & $a_9$ &$a_8$ &$a_7$ & $a_1$ &$a_2$ & $a_3$ &$a_4$ &$a_5$ &$a_6$ & $a_{15}$ &$a_{12}$ & $a_{14}$  & $a_{11}$ &$a_{10}$ &$a_{13}$\\ \hline
 
$a_5$ &  0 & $a_8$ & 0 &$a_8$ &$a_2$ & 0 &$a_2$ &$a_5$ & 0 &$a_5$ &$a_5$ &$a_5$ &$a_5$ &$a_8$ &$a_8$ &$a_2$\\ \hline
 
$a_6$ &  0 & $a_7$ &$a_8$ &$a_9$ &$a_3$ &$a_2$ & $a_1$ &$a_6$ &$a_5$ &$a_4$ &$a_{12}$ & $a_{15}$ &$a_{13}$  &$a_{10}$ &$a_{11}$ & $a_{14}$\\ \hline
 
$a_7$ &  0  &  $a_1$  & $a_2$ &$a_3$ &$a_4$ &$a_5$ &$a_6$ &$a_7$ &$a_8$ &$a_9$ &$a_{10}$ &$a_{11}$ &$a_{12}$ &$a_{13}$ & $a_{14}$ & $a_{15}$ \\ \hline 
 
$a_8$ &  0  & $a_2$  &  0 &$a_2$ &$a_5$ & 0 &$a_5$ &$a_8$ & 0 &$a_8$ &$a_8$ &$a_8$ &$a_5$ &$a_2$ &$a_2$ &$a_5$ \\ \hline 
 
$a_9$ &  0  & $a_3$  & $a_2$ & $a_1$ &$a_6$ &$a_5$ &$a_4$ &$a_9$ &$a_8$ &$a_7$ &$a_{11}$ &$a_{10}$ & $a_{15}$ & $a_{14}$ &$a_{13}$ &$a_{12}$ \\ \hline
 
$a_{10}$ &  0  & $a_{13}$  & $a_2$ & $a_{14}$ & $a_{15}$ &$a_5$ &$a_{12}$ &$a_{10}$ &$a_8$ &$a_{11}$ &$a_7$ &$a_9$ &$a_6$ & $a_1$ &$a_3$ &$a_4$ \\ \hline

$a_{11}$ &  0  &  $a_{14}$  & $a_2$ &$a_{13}$ &$a_{12}$ &$a_5$ & $a_{15}$ &$a_{11}$ &$a_8$ &$a_{10}$ &$a_9$ &$a_7$ &$a_4$ & $a_3$ & $a_1$ &$a_6$ \\ \hline 

$a_{12}$ & 0 &$a_{10}$ &$a_8$ &$a_{11}$ & $a_{14}$ &$a_5$ & $a_{13}$  &$a_{12}$ &$a_5$ & $a_{15}$ &$a_6$ &$a_4$ & $a_1$ &$a_7$ &$a_9$ &$a_3$ \\ \hline 

$a_{13}$ & 0 &$a_{12}$ & $a_5$ & $a_{15}$ &$a_{11}$ &$a_8$ &$a_{10}$ &$a_{13}$ &$a_2$ & $a_{14}$ & $a_1$ &$a_3$ &$a_7$ & $a_6$ &$a_4$ &$a_9$ \\ \hline 

$a_{14}$ & 0 & $a_{15}$ &$a_5$ &$a_{12}$ &$a_{10}$ &$a_8$ &$a_{11}$ & $a_{14}$ &$a_2$ &$a_{13}$ &$a_3$ & $a_1$ &$a_9$ &$a_4$ &$a_6$ &$a_7$ \\ \hline 

$a_{15}$ & 0 &$a_{11}$ &$a_8$ &$a_{10}$ &$a_{13}$ &$a_2$ & $a_{14}$ & $a_{15}$ &$a_5$ &$a_{12}$ &$a_4$ & $a_6$ &$a_3$ &$a_9$ &$a_7$ & $a_1$ \\ \hline
\end{tabular}\caption{Multiplication table of $W$}\label{p1_table2}
\end{table}
\end{center}
\vspace{-0.8cm}
We choose the element $1+X+Y+YX$ from $\mathbb{Z}_4+X\mathbb{Z}_4+Y\mathbb{Z}_4+YX\mathbb{Z}_4$ and denote it by $U$, i.e., $U = 1+X+Y+YX$. So $UW = \lbrace 0$, $3+Y3$, $2+Y2$, $1+Y$, $X+YX$, $X2+YX2$, $X3+YX3$, $1+X+Y+YX$, $2+X2+Y2+YX2$, $3+X3+Y3+YX3$, $3+X+Y3+YX$, $1+X3+Y+YX3$, $2+X+Y2+YX$, $3+X2+Y3+YX2$, $1+X2+Y+YX2$, $2+X3+Y2+YX3\rbrace$. This implies that $W \cap UW=\lbrace 0 \rbrace$, which inturn implies that $W+UW=\mathbb{Z}_4+X\mathbb{Z}_4+Y\mathbb{Z}_4+YX\mathbb{Z}_4$ as $W+UW$ is sub-ring of $\mathbb{Z}_4+X\mathbb{Z}_4+Y\mathbb{Z}_4+YX\mathbb{Z}_4$ and $\mid W+UW \mid =256$. Therefore 
$\mathbb{Z}_4+X\mathbb{Z}_4+Y\mathbb{Z}_4+YX\mathbb{Z}_4 = W+UW$.

%

Let $x^2+x+1$ be a basic irreducible polynomial over $\mathbb{Z}_4$. Then $\dfrac{\mathbb{Z}_4[x]}{\langle x^2+x+1 \rangle}$ is called the Galois extension ring of $\mathbb{Z}_4$ and is denoted by GR(4,2). If $w$ is a root of $x^2+x+1$ then $\dfrac{\mathbb{Z}_4[x]}{\langle x^2+x+1 \rangle} = GR(4,2) \cong \mathbb{Z}_4[w]$. 
\begin{lemma}\label{p1_lemma4}
The ring $W$ is isomorphic to the ring $\mathbb{Z}_4[w]$, i.e., $W\cong\mathbb{Z}_4[w]$.
\end{lemma}
\begin{proof}
We consider the mapping as follows: $0\longmapsto 0$, $X+Y3+YX3\longmapsto w$, $X2+Y2+YX2\longmapsto 2w$, $X3+Y+YX\longmapsto 3w$, $1+X+Y3\longmapsto 3w^2$, $2+X2+Y2\longmapsto 2w^2$, $3+X3+Y\longmapsto w^2$, $1+YX\longmapsto 1$, $2+YX2\longmapsto 2$, $3+YX3\longmapsto 3$, $1+X2+Y2+YX3\longmapsto 2w+1$, $3+X2+Y2+YX\longmapsto 2w+3$, $1+X3+Y+YX2\longmapsto 3w+1$, $2+X3+Y+YX3\longmapsto 3w+2$, $2+X+Y3+YX\longmapsto w+2$, $3+X+Y3+YX2\longmapsto w+3$. It is clear from Table \ref{p1_table2} that this map is a ring isomorphism. Therefore $W\cong\mathbb{Z}_4[w]$.
\end{proof}
\begin{thm}\label{p1_thm1}
The ring $M_2(\mathbb{Z}_4)$ is isomorphic to the ring $\mathbb{Z}_4[w]+U\mathbb{Z}_4[w]$, i.e., $M_2(\mathbb{Z}_4) \cong\mathbb{Z}_4[w]+U\mathbb{Z}_4[w]$.
\end{thm} 

\begin{proof}
We have 
\begin{eqnarray*} 
M_2(\mathbb{Z}_4) & \cong  & \mathbb{Z}_4+X\mathbb{Z}_4+Y\mathbb{Z}_4+YX\mathbb{Z}_4 ~~\mbox{from Lemma}~ \ref{p1_lemma1}\\
&\cong & W+UW \\
& \cong & \mathbb{Z}_4[w]+U\mathbb{Z}_4[w]  ~~\mbox{from Lemma}~ \ref{p1_lemma4}
\end{eqnarray*}
Therefore $M_2(\mathbb{Z}_4) \cong\mathbb{Z}_4[w]+U\mathbb{Z}_4[w]$.
\end{proof}
We notice here that the rings $W$ and $\mathbb{Z}_4[w]$ are commutative, however, their extensions, both $W+UW$ and $\mathbb{Z}_4[w]+U\mathbb{Z}_4[w]$ are non-commutative. Summarising the  above discussion, we have $\mathcal{R} \cong \mathbb{Z}_4[w]+U\mathbb{Z}_4[w]$, where $U^2=2U$, $U^3=0$, $2U=U2$ and $2U^2=0$.

We know that each element of $\mathbb{Z}_4$ has 2-adic representation $a+2b$, where $a,~b \in \mathbb{Z}_2$, so is $\mathbb{Z}_4[w]$. Now we define a Gray map on $\mathcal{R}$. For this, first we define a mapping $\mathcal{R}$ to $\mathbb{Z}^2_4[w]$, and then define a mapping $\mathbb{Z}^2_4[w]$ to $\mathbb{F}^4_4$ so that the Gray map is 

\begin{equation*}
\begin{array}{cccc}
\Phi: & \mathcal{R}         & \longrightarrow & \mathbb{F}_4^4\\
      & a+2b+Uc+U2d        & \longmapsto & (d,~ c+d,~ b+d,~ a+b+c+d),
\end{array}
\end{equation*}
where $a, b, c, d\in \mathbb{F}_4$. This map can easily be extended to $\mathcal{R}^n$ component wise. The Hamming weight $w_H$ of $x \in \mathbb{F}_4^n$ is defined as the number of non-zero coordinates of $x$. For $x=a+2b+Uc+U2d \in \mathcal{R}^n$, we define  the Lee weight of $x$, denoted by $w_L(x)$,  as $$w_L(x)=w_H(d)+w_H(d+c) + w_H(d+b)+w_H(a+b+c+d).$$ For any $x$, $y \in \mathcal{R}^n$, the Lee distance $d_L(x, y)$  between $x$ and $y$ is the Lee weight of $x-y$, i.e., $d_L(x,y)=w_L(x-y)$. 
A linear code $C$ of length $n$ over $\mathcal{R}$ is  an $\mathcal{R}$-submodule of $\mathcal{R}^n$. $C$ is said to be a {\it{free code}} if $C$ has a $\mathcal{R}$-basis. We define the \emph{rank} of a code $C$ as the minimum number of generators for $C$. The Lee distance of $C$ is denoted by $d_L(C)$ and is defined by $d_L(C)=min \lbrace  w_L(c)=\sum^{n-1}_{i=0 }w_L(c_i)\vert c=(c_0, c_1,\dots, c_{n-1})\in C\rbrace$. From the above discussion, we can easily verify the following theorem.

\begin{thm}\label{p1_thm2}
If $C$ is a linear code over $\mathcal{R}$ of length $n$, size $M$ with Lee distance $d_L$, then $\Phi(C)$ is a code  of length $4n$ over $\mathbb{F}_4$, size $M$.
\end{thm}

\section{Cyclic codes over $M_2(\mathbb{Z}_4)$}

Let $\tau$ be the standard cyclic shift operator on $\mathcal{R}^n$. A linear code $C$ of length $n$ over $\mathcal{R}$ is cyclic if $\tau(c) \in C$ whenever $c \in C$, \emph{i.e.,} if $(c_0, c_1, \ldots, c_{n-1}) \in C$, then $(c_{n-1}, c_0, c_1, \ldots, c_{n-2}) \in C$. As usual, in the polynomial representation, a cyclic code of length $n$ over $\mathcal{R}$ is an ideal of $\frac{\mathcal{R}[x]}{\left\langle x^n-1\right\rangle}$. Note here that $\frac{\mathcal{R}[x]}{\left\langle x^n-1\right\rangle}$ is a ring.\\

\begin{thm}\label{thm2.2}
A linear code $C = C_1 + UC_2$ of length $n$ over $\mathcal{R}$ is cyclic if and only if $C_1$, $C_2$ are cyclic codes of length $n$ over $\mathbb{Z}_4[w]$.
\end{thm}
\begin{proof}
Let $c_1+Uc_2 \in C$, where $c_1 \in C_1$ and $c_2 \in C_2$. Then $\tau(c_1+Uc_2) = \tau(c_1) +U\tau(c_2) \in C$, since $C$ is cyclic and $\tau$ is a linear map. So, $\tau(c_1) \in C_1$ and $\tau(c_2) \in C_2$. Therefore $C_1, C_2$ are cyclic codes. Conversely if $C_1$, $C_2$ are cyclic codes, then for any $c_1+Uc_2 \in C$, where $c_1 \in C_1$ and $c_2 \in C_2$, we have $\tau(c_1) \in C_1$ and $\tau(c_2) \in C_2$, and so, $\tau(c_1+Uc_2) = \tau(c_1) +U\tau(c_2) \in C$. Hence $C$ is cyclic.
\end{proof}

We assume that $n$ is odd for  the rest of this paper. Let $\mathcal{R}[x]$ be the ring of polynomials over the ring $\mathcal{R}$. 
We define a mapping  
\begin{equation*}
\begin{array}{cccc}
\mu: & \mathcal{R}[x]          & \longrightarrow & \mathbb{F}_4[x] \\
     & \sum_{i=0}^n a_ix^i    & \longmapsto & \sum_{i=0}^n \mu(a_i)x^i,
\end{array}
\end{equation*}
where $\mu(a_i)$ denote reduction of modulo $2$ and $U$.


A polynomial $f \in \mathcal{R}[x]$ is called {\it{basic irreducible polynomial}} if $\mu(f)$ is irreducible over $\mathbb{F}_4$. Two polynomials $f(x), g(x) \in \mathcal{R}[x]$ are said to be \emph{coprime} if there exist $a(x), b(x) \in \mathcal{R}[x]$ such that
\[a(x)f(x) + b(x)g(x) = 1~.\]

The polynomial $x^n-1$ factorizes uniquely into pairwise coprime irreducible polynomials over $\mathbb{F}_4$. Let $x^n-1=f_1f_2f_3\cdots f_m$, where $f_i$'s are irreducible polynomials over $\mathbb{F}_4$.

\begin{lemma} 
Let $f_i$ be a basic irreducible polynomial over $\mathcal{R}$. Then $\dfrac{\mathcal{R}[x]}{\langle f_i \rangle}$ is not a ring but a right module over $\mathcal{R}$.
\end{lemma}
\begin{proof}
Since $\langle f_i \rangle$ is not two sided ideal of $\mathcal{R}[x]$ so $\dfrac{\mathcal{R}[x]}{\langle f_i \rangle }$ is not a ring for $1 \leq i \leq m$. Then each $\dfrac{\mathcal{R}[x]}{\langle f_i \rangle }$ is a right $\mathcal{R}$-module.
\end{proof}

We need a non-commutative analogue of the Chinese Remainder Theorem for modules.
 
\begin{lemma} 
Let $n$ be an odd integer. Then $$\dfrac{\mathcal{R}[x]}{\langle x^n-1 \rangle} = \bigoplus^m_1 \dfrac{\mathcal{R}[x]}{\langle f_i \rangle}.$$
\end{lemma}
\begin{proof}
The proof follows from \cite{Oggier2012}, \cite{Wisbauer1991}.
\end{proof}
  
\begin{thm}\label{p1_thm3}
If $f$ be an irreducible polynomial in $\mathbb{F}_4[x]$, then the right $\mathcal{R}$-modules of $\dfrac{\mathcal{R}[x]}{\langle f \rangle}$ are $$\langle 0 \rangle, \langle 1+\langle f \rangle \rangle, \langle U+\langle f \rangle \rangle, \langle 2U+\langle f \rangle \rangle, \langle (2+Um_{f}) + \langle f \rangle \rangle, \langle 2+\langle f \rangle \rangle, \langle \langle 2, U \rangle+\langle f \rangle \rangle,$$  where $m_{f}$ is an unit in $\dfrac{\mathbb{F}_4[x]}{\langle f \rangle}$.
\end{thm}

\begin{proof}
Let $I$ be a non-zero right sub-module of $\dfrac{\mathcal{R}[x]}{\langle f \rangle}$ and $g(x) \in \mathcal{R}[x]$ such that $g(x) + \langle f \rangle \in I$ but $g(x)\notin \langle f \rangle$. If $gcd(\mu(g(x))$, $f(x))=1$ then $g$ is invertible $\pmod f$. So $I=\langle 1 + \langle f \rangle \rangle =\dfrac{\mathcal{R}[x]}{\langle f \rangle}$. If $gcd(\mu(g(x)), f(x)) = f(x)$ then there exit $g_1(x)$, $g_2(x)$, $g_3(x)$, $g_4(x)\in\mathbb{F}_4[x]$ such that $g(x) = g_1(x) + Ug_2(x) + 2g_3(x) + 2Ug_4(x)$ with $gcd((g_1(x)), f(x))=f(x)$ then $g(x) +\langle f \rangle = Ug_2(x)+2g_3(x)+2Ug_4(x)+\langle f \rangle$. If $gcd(g_2(x), f(x))=f(x)$ then  $g(x)+\langle f \rangle =2g_3(x)+2Ug_4(x)+\langle f \rangle$. It follows that $I=\langle 2+\langle f \rangle \rangle$. If $gcd(g_3(x), f(x))=f(x)$, implies that $I=\langle 2U+\langle f \rangle \rangle$. Also if $gcd(g_2(x), f(x))=1$, then there exits $g^{-1}_2(x)\in \mathbb{F}_4[x]$ such that $g_2(x)g^{-1}_2(x)\equiv 1 \pmod f$. Therefore $2U=2g(x)g^{-1}_2(x)$. Hence $2U+\langle f \rangle =2g(x)g^{-1}_2(x)+\langle f \rangle\in I$. It follows that $ Ug_2(x)+2g_3(x)+\langle f \rangle = g(x) + 2Ug_4(x)+\langle f \rangle\in I$. If $gcd(g_3(x), f(x))=f(x), $ then $I=\langle U+\langle f \rangle \rangle$. Otherwise $gcd(g_3(x), f(x))=1$, then $g^{-1}_3(x)\in \mathbb{F}_4[x]$ such that $g_3(x)g^{-1}_3(x)\equiv 1 \pmod f$. Hence $2+Ug_2(x)g^{-1}_3(x)+\langle f \rangle \in I$, i.e., $\langle 2+Um_{f}+\langle f \rangle \rangle  =I$, where $m_{f}=g_2(x)g_3^{-1}(x)$ is an unit in $\dfrac{\mathbb{F}_4[x]}{\langle f \rangle}$. Since $gcd(g_2(x), f(x))=1,~ gcd(g_3(x), f(x))=1$ then there exit $a_1(x)$, $a_2(x)$, $b_1(x)$, $b_2(x)\in\mathbb{F}_4[x]$ such that $g_2(x)a_1(x)+f(x)a_2(x)=1$, $g_3(x)b_1(x)+f(x)b_2(x)=1$. Therefore $$Ub_1(x)+\langle f \rangle = (Ug_2(x) + \langle f \rangle)(a_1(x)b_1(x)+\langle f \rangle)$$
$$2a_1(x)+\langle f \rangle=(Ug_3(x)+\langle f \rangle)(a_1(x)b_1(x)+\langle f \rangle)$$     $$Ub_1(x)+2a_1(x)+\langle f \rangle=(Ug_2(x)+2g_3(x)\langle f \rangle)(a_1(x)b_1(x)+\langle f \rangle)$$ It follows that $I=\langle \langle U, 2 \rangle + \langle  f \rangle \rangle$.
\end{proof}

\begin{thm}\label{p1_thm4}
Let $x^n-1=f_1f_2f_3 \cdots f_m$, where $f_i$'s are monic basic irreducible pairwise coprime polynomials in $\mathcal{R}[x]$. Let $\hat{f}_i=\dfrac{x^n-1}{f_i}$. Then any ideal in $\dfrac{\mathcal{R}[X]}{\left\langle x^n-1\right\rangle}$ is the sum of the right sub-modules: $\left\langle \hat{f}_i+\left\langle x^n-1\right\rangle \right\rangle$, $\left\langle 2\hat{f}_i+\left\langle x^n-1\right\rangle \right\rangle$, $\left\langle U\hat{f}_i+\left\langle x^n-1\right\rangle \right\rangle$, $\left\langle 2U\hat{f}_i+\left\langle x^n-1\right\rangle \right\rangle$, $\left\langle (2+Um_f)\hat{f}_i+\left\langle x^n-1\right\rangle \right\rangle$, $\left\langle \left\langle2, U \right\rangle \hat{f}_i+\left\langle x^n-1\right\rangle \right\rangle $, where $m_{f}$ is an unit in $\dfrac{\mathbb{F}_4[x]}{\left\langle f \right\rangle}$. 
\end{thm}

\begin{proof}
It follows from the Chinese Remainder Theorem for modules and the right $\mathcal{R}$-modules of the $\dfrac{\mathcal{R}[x]}{\left\langle f\right\rangle}$.
\end{proof}

\begin{thm}\label{p1_thm5}
Let $C$ be a cyclic code of length $n$ over $\mathcal{R}$. Then there exists a family of pairwise monic polynomials $F_0, F_1, \dots, F_6\in \mathbb{F}_4[x]$ such that  $F_0F_1\cdots F_6=x^n-1$ and $C=\left\langle \hat{F}_1  \right\rangle$ $\oplus$ $\left\langle U\hat{F}_2  \right\rangle $ $\oplus $  $ \left\langle 2\hat{F}_3  \right\rangle $ $\oplus $ $\left\langle 2U\hat{F}_4  \right\rangle $ $\oplus $  $\left\langle (2+Um_f)\hat{F}_5  \right\rangle $ $\oplus $  $ \left\langle \left\langle 2, U\right\rangle  \hat{F}_6  \right\rangle$, where $m_{f}$ is an unit in $\dfrac{\mathbb{F}_4[x]}{\left\langle f \right\rangle}$. Moreover $\vert C \vert = 4^\alpha$, where $\alpha = 4 degF_1+2 degF_2+2 degF_3+degF_4+2 degF_5+3 degF_6$.
\end{thm}

\begin{proof}
First part follows from a similar argument of Theorem 2 in \cite{Luo2017}. Now we compute $|C|$. We know that $C=\left\langle \hat{F}_1  \right\rangle $ $\oplus $   $\left\langle U\hat{F}_2  \right\rangle $ $\oplus $  $ \left\langle 2\hat{F}_3  \right\rangle $ $\oplus $ $\left\langle 2U\hat{F}_4  \right\rangle $ $\oplus $  $\left\langle (2+Um_f)\hat{F}_5  \right\rangle $ $\oplus $  $ \left\langle \left\langle 2,U\right\rangle  \hat{F}_6  \right\rangle$, which implies that $\vert C \vert  =  \mid\hat{F}_{1}\mid \cdot \mid U\hat{F}_{2}\mid  \cdot \mid 2\hat{F}_{3}\mid  \cdot \mid 2U\hat{F}_{4}\mid  \cdot \mid (2+Um_f)\hat{F}_{5}\mid  \cdot \mid \langle2, U\rangle\hat{F}_{6}\mid$. The rest follows from the fact that 
$\mid\hat{F}_{1}\mid=4^{4degF_1}$,  $|U\hat{F}_{2}|=4^{2degF_2}$, $ \mid 2\hat{F}_{3}\mid=4^{2degF_3}$, $\mid 2U\hat{F}_{4}\mid=4^{degF_4}$, $\mid (2+Um_f)\hat{F}_{5}\mid=4^{2degF_5}$,  $\mid \langle2, U\rangle\hat{F}_{6}\mid=4^{3degF_6}$.
 \end{proof}
 
\begin{thm}\label{p1_thm6}
Let $C$ be a cyclic code of length $n$ over $\mathcal{R}$ with $C=\left\langle \hat{F}_1  \right\rangle \oplus \left\langle U\hat{F}_2  \right\rangle \oplus  \left\langle 2\hat{F}_3  \right\rangle \oplus \left\langle 2U\hat{F}_4  \right\rangle \oplus \left\langle (2+Um_f)\hat{F}_5  \right\rangle \oplus  \left\langle \left\langle 2, U\right\rangle  \hat{F}_6  \right\rangle$, where $m_{f}$ is an unit in $\dfrac{\mathbb{F}_4[x]}{\left\langle f \right\rangle}$ and $F=\hat{F}_1+U\hat{F}_2+2\hat{F}_3+2U\hat{F}_4+(2+Um_f)\hat{F}_5+\left\langle 2, U\right\rangle \hat{F}_6$. Then $C=\left\langle  F \right\rangle$.
\end{thm}
\begin{proof}
For any two distinct $i, j$, $0 \leq i, j \leq 6$, we have $(x^n-1)\vert \hat{F}_i\hat{F}_j$. So $\hat{F}_i\hat{F}_j=0$. Also for any $i$ with $0\leq i \leq 6$, $ F_i ,\hat{F}_i $ are coprime and $ F_i \hat{F}_i=0$. Since  $ F_i, \hat{F}_i $ are coprime, there exist $a_i, b_i$ such that $(a_1F_1+b_1\hat{F}_1)(a_2F_2+b_2\hat{F}_2)(a_3F_3+b_3\hat{F}_3)(a_4F_4+b_4\hat{F}_4)(a_5F_5+b_5\hat{F}_5)=1$. This implies that $a_1F_1a_2F_2a_3F_3a_4F_4a_5F_5 + b_1\hat{F}_1a_2F_2a_3F_3a_4F_4a_5F_5 + a_1F_1b_2\hat{F}_2a_3F_3a_4F_4a_5F_5 + a_1F_1a_2F_2b_3\hat{F}_3a_4F_4a_5F_5 +a_1F_1a_2F_2a_3F_3b_4\hat{F}_4a_5F_5 +a_1F_1a_2F_2a_3F_3a_4F_4b_5\hat{F}_5 =1$. On multiplying both side by $\hat{F}_6$, we obtain $\hat{F}_6a_1F_1a_2F_2a_3F_3a_4F_4a_5F_5=\hat{F}_6$.  We have $F=\hat{F}_1+U\hat{F}_2+2\hat{F}_3+2U\hat{F}_4+(2+Um_f)\hat{F}_5+\left\langle 2, U\right\rangle \hat{F}_6$.  It follows that  $Fa_1F_1a_2F_2a_3F_3a_4F_4a_5F_5    =\left\langle 2, U\right\rangle \hat{F}_6a_1F_1a_2F_2a_3F_3a_4F_4a_5F_5,$  which inturn implies that $Fa_1F_1a_2F_2a_3F_3a_4F_4a_5F_5=\langle 2, U\rangle \hat{F}_6$. Hence $ \langle 2, U\rangle \hat{F}_6\in \langle F\rangle$. Continuing this process, we obtain $\hat{F}_1,~ U\hat{F}_2,~ 2\hat{F}_3,~ 2U\hat{F}_4, ~(2+Um_f)\hat{F}_5,~ \left\langle 2, U\right\rangle \hat{F}_6\in \langle F \rangle$.  Consequently $C=\langle F \rangle$.
\end{proof}
Let us denote $R=\dfrac{\mathbb{F}_4[x]}{\left\langle x^n-1 \right\rangle}$.

\begin{thm}\label{p1_thm7}
Let $C$ be a cyclic code of length $n$ over $\mathcal{R}$. Then there exists a family of polynomials $F, G, H, Q, T\in \mathbb{F}_4[x] $ which are divisors of $x^n-1$ such that $C=\left\langle F \right\rangle_{R} \oplus U\left\langle G \right\rangle_{R} \oplus 2\left\langle H \right\rangle_{R}\oplus  2U\left\langle Q \right\rangle_{R}\oplus (2+Um_f) \left\langle T \right\rangle_{R}$, where $m_{f}$ is an unit in $\dfrac{\mathbb{F}_4[x]}{\left\langle f \right\rangle}$. Moreover $\vert C \vert=4^{5n-(degF+degG+degH+degQ+degT)}$. 
\end{thm} 
\begin{proof}
A similar argument as in \cite{Luo2017}. 
\end{proof}

\section{Self-dual cyclic codes over $M_2(\mathbb{Z}_4)$}

For given $\textbf{x}=(x_1, x_2, \dots, x_n)$, $\textbf{y}=(y_1, y_2, \dots, y_n)\in \mathcal{R}^n$, the Euclidean scalar product (or dot product) of $\textbf{x, y}$ is $\textbf{x}\cdot \textbf{y} = x_1y_1+x_2y_2+ \cdots+x_ny_n \pmod 4$.  Two vectors $\textbf{x}$ and $\textbf{y}$ in $\mathcal{R}^n$ are called orthogonal if $\textbf{x}\cdot \textbf{y}=0$. 
For a linear code $C $ over $\mathcal{R}$, its dual code $C^\perp$ is the set of words over $\mathcal{R}$ that are orthogonal to all codewords of $C$, i.e., $C^\perp=\left\lbrace \textbf{x} \in \mathcal{R}^n \mid \textbf{x} \cdot \textbf{y}=0, \forall y\in C \right\rbrace$.  A code $C$ is called self-orthogonal if $C \subset C^{\perp}$ and self-dual if $C=C^{\perp}$.


Let $f(x)=a_0+a_1x+ \cdots +a_{k-1}x^{k-1} + a_{k}x^{k}$ be a polynomial of degree $k$ with $a_{k}\neq 0$, $a_{0}\neq 0$. The reciprocal $f^\ast(x)$ of $f(x)$ is defined by $$f^\ast(x)=a_0^{-1}x^kf(x^{-1}).$$

\begin{thm}\label{p1_thm8}
Let $C$ be a cyclic code of length $n$ over $\mathcal{R}$ with $C=\left\langle \hat{F}_1  \right\rangle $ $\oplus $   $\left\langle U\hat{F}_2  \right\rangle $ $\oplus $  $ \left\langle 2\hat{F}_3  \right\rangle $ $\oplus $ $\left\langle 2U\hat{F}_4  \right\rangle $ $\oplus $  $\left\langle (2+Um_f)\hat{F}_5  \right\rangle $ $\oplus $  $ \left\langle \left\langle 2, U\right\rangle  \hat{F}_6  \right\rangle$, where $m_{f}$ is an unit in $\dfrac{\mathbb{F}_4[x]}{\left\langle f \right\rangle}$. Then $C^\perp=\left\langle \hat{F}^\ast_0  \right\rangle $ $\oplus $   $\left\langle U\hat{F}^\ast_2  \right\rangle $ $\oplus $  $ \left\langle 2\hat{F}^\ast_3  \right\rangle $ $\oplus $ $\left\langle 2U\hat{F}^\ast_6  \right\rangle $ $\oplus $  $\left\langle (2+Um_f)\hat{F}^\ast_5  \right\rangle $ $\oplus $  $ \left\langle \left\langle 2, U\right\rangle  \hat{F}^\ast_4  \right\rangle$ and $\mid C^\perp\mid=4^{4degF_0+2degF_2+2degF_3+3degF_4+2degF_5+degF_6}$.
\end{thm}

\begin{proof}
From the Theorem \ref{p1_thm5}, $\mid C \mid = 4^{4degF_1 + 2degF_2 + 2degF_3 + degF_4 + 2degF_5 + 3degF_6}$.  Since $\mid C \mid \mid C^\perp \mid=4^{4n}$ and $n=degF_1+degF_2+degF_3+degF_4+degF_5+degF_6$, so $\mid C^\perp \mid=4^{4degF_0+2degF_2+2degF_3+3degF_4+2degF_5+degF_6}$.  

We denote $C^\ast=\left\langle \hat{F}^\ast_0  \right\rangle $ $\oplus $   $\left\langle U\hat{F}^\ast_2  \right\rangle $ $\oplus $  $ \left\langle 2\hat{F}^\ast_3  \right\rangle $ $\oplus $ $\left\langle 2U\hat{F}^\ast_6  \right\rangle $ $\oplus $  $\left\langle (2+Um_f)\hat{F}^\ast_5  \right\rangle $ $\oplus $  $ \left\langle \left\langle 2, U\right\rangle  \hat{F}^\ast_4  \right\rangle  $. For $i, j$, $0 \leq i, j \leq 6$, if   $i+1=7-j+1$, i.e., $i=7-j,$ we can see that  $\hat{F}_{i+1}\hat{F}^\ast_{7-i+1}=0$. If $i+1\neq7-j+1$, i.e., $i\neq7-j,$ then we have $ x^n-1\mid\hat{F}_{i+1}\hat{F}^\ast_{7-i+1}$. Then it follows that $\hat{F}_{i+1}\hat{F}^\ast_{7-i+1}=0$. Therefore $C^\ast \subseteq C^\perp$. Note that
$\mid\hat{F}^\ast_{0}\mid=4^{4degF_0}$,    $\mid  U\hat{F}^\ast_{2}\mid=4^{2degF_2}$, $\mid 2\hat{F}^\ast_{3}\mid=4^{2degF_3}$, 
$\mid 2U\hat{F}^\ast_{6}\mid=4^{degF_6}$, $\mid (2+Um_f)\hat{F}^\ast_{5}\mid=4^{2degF_5}$, $\mid \langle2, U\rangle\hat{F}^\ast_{4}\mid=4^{degF_4}$. Hence $|C^\ast| = 4^{4degF_0+2degF_2+2degF_3+3degF_4+2degF_5+degF_6} = |C^\perp|$. Consequently $C^\ast = C^\perp$.
\end{proof}

\begin{thm}\label{p1_thm9}
Let $C$ be a cyclic code of length $n$ over $\mathcal{R}$ with $C$ and $C^\perp$ defined as in Theorem \ref{p1_thm8},
and $F^\ast=\hat{F}^\ast_0+U\hat{F}^\ast_2+2\hat{F}_3+2U\hat{F}^\ast_6+(2+Um_f)\hat{F}^\ast_5+\left\langle 2, U\right\rangle \hat{F}^\ast_4$. Then $C^\perp=\left\langle  F^\ast \right\rangle$.
\end{thm}

\begin{proof}
The result follows from a similar argument as in the proof of Theorem \ref{p1_thm6} as $\hat{F}^\ast_i\hat{F}^\ast_j=0$  and $\hat{F}^\ast_i, F^\ast_j$ are coprime for any $i, j$, $0 \leq i,j \leq 6$.
\end{proof} 

\begin{thm}\label{p1_thm10}
Let $C$ be a cyclic code of length $n$ over $\mathcal{R}$. Then there exists a family of polynomials $F^\ast, G^\ast, H^\ast, Q^\ast, T^\ast\in \mathbb{F}_4[x] $ which are divisors of $x^n-1$ such that $C^\perp=\left\langle F^\ast \right\rangle_{R} \oplus U\left\langle G^\ast \right\rangle_{R} \oplus 2\left\langle H^\ast \right\rangle_{R}\oplus  2U\left\langle Q^\ast \right\rangle_{R}\oplus (2+Um_f) \left\langle T^\ast \right\rangle_{R}$, where $m_{f}$ is an unit in $\dfrac{\mathbb{F}_4[x]}{\left\langle f \right\rangle}$. Moreover $\vert C^\perp \vert = 4^{5n-(degF^\ast+degG^\ast+degH^\ast+degQ^\ast+degT^\ast)}$. 
 \end{thm} 
 
\begin{proof}
Follows fromTheorem \ref{p1_thm7}. 
\end{proof}

We now prove the main result of this section, a condition for a cyclic code $C$ over $\mathcal{R}$ to be self-dual. From Theorem \ref{p1_thm6} and Theorem \ref{p1_thm9}, we can see that a cyclic codes $C$ is self-dual if and only if $F=F^\ast$. 
This implies  that 
$$ \hat{F}_1=\hat{F}^\ast_0, ~~ \hat{F}_2=\hat{F}^\ast_2,~~ \hat{F}_3=\hat{F}^\ast_3,~~\hat{F}_4=\hat{F}^\ast_6,~~\hat{F}_5=\hat{F}^\ast_5,~~ \hat{F}_6=\hat{F}^\ast_4.$$
Again since $\hat{F}_i=\dfrac{x^n-1}{F_i}$, $\hat{F}^\ast_j=\dfrac{x^n-1}{F^\ast_j}$ and $\hat{F}_i=\hat{F}^\ast_j$, we have $F_i=F^\ast_j$. Hence proved the following results.

\begin{thm}\label{p1_thm11}
Let $C$ be a cyclic code of length $n$ over $\mathcal{R}$ with $C=\left\langle \hat{F}_1  \right\rangle $ $\oplus $   $\left\langle U\hat{F}_2  \right\rangle $ $\oplus $  $ \left\langle 2\hat{F}_3  \right\rangle $ $\oplus $ $\left\langle 2U\hat{F}_4  \right\rangle $ $\oplus $  $\left\langle (2+Um_f)\hat{F}_5  \right\rangle $ $\oplus $  $ \left\langle \left\langle 2 ,U\right\rangle  \hat{F}_6  \right\rangle$, where $m_{f}$ is an unit in $\dfrac{\mathbb{F}_4[x]}{\left\langle f \right\rangle}$. Then $C$ is self-dual code 
if and only if $F_1=F^\ast_0,~~ F_2=F^\ast_2,~~ F_3=F^\ast_3,~~ F_4=F^\ast_6,~~ F_5=F^\ast_5.$
\end{thm}
\begin{thm}\label{p1_thm12}
Let $C$ be a cyclic code of length $n$ over $\mathcal{R}$ with $C=\left\langle F \right\rangle_{R} \oplus U\left\langle F \right\rangle_{R} \oplus 2\left\langle F \right\rangle_{R}\oplus  2U\left\langle F \right\rangle_{R}\oplus (2+Um_f) \left\langle F \right\rangle_{R},$  where $m_{f}$ is an unit in $\dfrac{\mathbb{F}_4[x]}{\left\langle f \right\rangle}$. Then $C$ is self-dual code if and only if $F=F^\ast,~~ G=G^\ast,~~ H=H^\ast,~~ Q=Q^\ast,~~ T=T^\ast.$
\end{thm}

\section{Hermitian Self-dual cyclic codes over $M_2(\mathbb{Z}_4)$}
For any two codewords $\textbf{x}=(x_1, x_2, \dots, x_n)$, $\textbf{y}=(y_1, y_2, \dots, y_n)\in \mathcal{R}^n$, the Hermitian inner product is defined as $$\langle\textbf{x},\textbf{y}\rangle=\textbf{x} \cdot {\bar{\textbf{y}}}=x_1\bar{y_1}+x_2\bar{y_2}+ \cdots+x_n\bar{y_n},$$ where  `` $\bar{~}$ " called conjugation, for example, $\bar{0}=0$, $\bar{1}=1$, $\bar{w}=w^2$, $\bar{w^2}=w$.  The Hermitian dual of $C$, denoted by $C^{\perp_{H}}$,  is define as $$C^{\perp_{H}}=\left\lbrace\textbf{x} \in \mathcal{R}^n \mid \langle\textbf{x},\textbf{y}\rangle=0,\forall y\in C \right\rbrace .$$  We can see that $\bar{C}^\perp=C^{\perp_{H}}$. As usual $C$ is called Hermitian self-orthogonal and  Hermitian self-dual if $C\subseteq C^{\perp_{H}}$ and $C=C^{\perp_{H}}$, respectively. 

Let $f(x)=a_0+a_1x+ \cdots +a_{k-1}x^{k-1} + a_{k}x^{k}$ be a polynomial of degree $k$ with $a_{k}\neq 0$, $a_{0}\neq 0$. The reciprocal $f^\ast(x)$ of $f(x)$ is defined by $$f^\ast(x)=a_0^{-1}x^kf(x^{-1}).$$
Denote $\bar{f}(x)=a^2_0+a^2_1x+ \cdots +a^2_{k-1}x^{k-1} + a^2_{k}x^{k}$. It is easy to check that two operations $\ast$ and $\bar{•}$ are commutative, i.e., $\bar{(f^\ast)}(x)=(\bar{f})^\ast(x)$. All the theorems proved in previous section are true with respect to Hermitian inner product as well. So state them here  without proofs.
\begin{thm}\label{p1_thm13}
Let $C$ be a cyclic code of length $n$ over $\mathcal{R}$ with $C=\left\langle \hat{F}_1  \right\rangle $ $\oplus $   $\left\langle U\hat{F}_2  \right\rangle $ $\oplus $  $ \left\langle 2\hat{F}_3  \right\rangle $ $\oplus $ $\left\langle 2U\hat{F}_4  \right\rangle $ $\oplus $  $\left\langle (2+Um_f)\hat{F}_5  \right\rangle $ $\oplus $  $ \left\langle \left\langle 2, U\right\rangle  \hat{F}_6  \right\rangle$, where $m_{f}$ is an unit in $\dfrac{\mathbb{F}_4[x]}{\left\langle f \right\rangle}$. Then $C^{\perp_{H}}=\left\langle \hat{\bar{F}}^\ast_0  \right\rangle $ $\oplus $   $\left\langle U\hat{\bar{F}}^\ast_2  \right\rangle $ $\oplus $  $ \left\langle 2\hat{\bar{F}}^\ast_3  \right\rangle $ $\oplus $ $\left\langle 2U\hat{\bar{F}}^\ast_6  \right\rangle $ $\oplus $  $\left\langle (2+Um_f)\hat{\bar{F}}^\ast_5  \right\rangle $ $\oplus $  $ \left\langle \left\langle 2, U\right\rangle  \hat{\bar{F}}^\ast_4  \right\rangle$ and $\mid C^{\perp_{H}}\mid=4^{4degF_0+2degF_2+2degF_3+3degF_4+2degF_5+degF_6}$.
\end{thm}
\begin{thm}\label{p1_thm14}
Let $C$ be a cyclic code of length $n$ over $\mathcal{R}$ with $C$ and $C^\perp$ defined as in Theorem \ref{p1_thm13}, and $\bar{F}^\ast=\hat{\bar{F}}^\ast_0+U\hat{\bar{F}}^\ast_2+2\hat{\bar{F}}_3+2U\hat{\bar{F}}^\ast_6+(2+Um_f)\hat{\bar{F}}^\ast_5+\left\langle 2, U\right\rangle \hat{\bar{F}}^\ast_4$. Then $C^{\perp_{H}}=\left\langle \bar{F}^\ast \right\rangle$.
\end{thm}

\begin{thm}\label{p1_thm14 a}
Let $C$ be a cyclic code of length $n$ over $\mathcal{R}$. Then there exists a family of polynomials $\bar{F}^\ast, \bar{G}^\ast, \bar{H}^\ast, \bar{Q}^\ast, \bar{T}^\ast\in \mathbb{F}_4[x] $ which are divisors of $x^n-1$ such that $C^{\perp_{H}}=\left\langle \bar{F}^\ast \right\rangle_{R} \oplus U\left\langle \bar{G}^\ast \right\rangle_{R} \oplus 2\left\langle \bar{H}^\ast \right\rangle_{R}\oplus  2U\left\langle \bar{Q}^\ast \right\rangle_{R}\oplus (2+Um_f) \left\langle \bar{T}^\ast \right\rangle_{R}$, where $m_{f}$ is an unit in $\dfrac{\mathbb{F}_4[x]}{\left\langle f \right\rangle}$. Moreover $\vert C^{\perp_{H}} \vert = 4^{5n-(degF^\ast+degG^\ast+degH^\ast+degQ^\ast+degT^\ast)}$. 
 \end{thm} 
%

 In \cite[Theoem 2]{Alahmadi2013}, the authors have claimed that Hermitian Self-dual codes do not exist over $M_2(\mathbb{Z}_2)$ but which is not true. In this paper, we present a condition for a cyclic code over $M_2(\mathbb{Z}_2)$ to be Hermitian Self-dual and also demonstrate  the same with an example. We also generalise the same to cyclic codes over $M_2(\mathbb{Z}_4)$.

\begin{thm}\label{p1_thm15}
Let $\mathcal{C}=<fh,~ ufg>$ be a cyclic code of length $n$ over $M_2(\mathbb{Z}_2)$  with $x^n-1=fgh$.  Then $\mathcal{C}$ is  Hermitian self-dual code  if and only if $f=\hat{\bar{g}}^\ast$ and $h=\hat{\bar{h}}^\ast$.
\end{thm}
In \cite[ $\mathsection$6]{Alahmadi2013}, authors have claimed that there does not exist a nontrivial self-dual cyclic codes of length $5$. A contrary example is the following:

\begin{example}
The factorization of $x^5-1$ is $(x-1)(x^2+wx+1)(x^2+w^2x+1)$ over $\mathbb{F}_4$. Let $f_1=(x-1)$, $f_2=(x^2+wx+1)$ and $f_3=(x^2+w^2x+1)$, then $f_1=\bar{f}^\ast_1$, $f_2=\bar{f}^\ast_3$ and $f_3=\bar{f}^\ast_2$. The following cyclic codes of length $5$ over $M_2(\mathbb{Z}_2)$ are self-dual (Hermitian) codes and their Gray image $\Phi(C)$ has parameters $[10,5,4]$ over $\mathbb{F}_4$.
$$\langle f_1f_2, \quad uf_2f_3 \rangle, \qquad \langle f_1f_3, \quad  uf_2f_3 \rangle.$$
\end{example}

\begin{thm}\label{p1_thm15}
Let $C$ be a cyclic code of length $n$ over $\mathcal{R}$ with $C=\left\langle \hat{F}_1  \right\rangle $ $\oplus $   $\left\langle U\hat{F}_2  \right\rangle $ $\oplus $  $ \left\langle 2\hat{F}_3  \right\rangle $ $\oplus $ $\left\langle 2U\hat{F}_4  \right\rangle $ $\oplus $  $\left\langle (2+Um_f)\hat{F}_5  \right\rangle $ $\oplus $  $ \left\langle \left\langle 2 ,U\right\rangle  \hat{F}_6  \right\rangle$, where $m_{f}$ is an unit in $\dfrac{\mathbb{F}_4[x]}{\left\langle f \right\rangle}$. Then $C$ is Hermitian self-dual code  if and only if $$\hat{F}_1=\hat{\bar{F}}^\ast_0,~~\hat{F}_2=\hat{\bar{F}}^\ast_2,~~\hat{F}_3=\hat{\bar{F}}^\ast_3,~~\hat{F}_4=\hat{\bar{F}}^\ast_6,~~\hat{F}_5=\hat{\bar{F}}^\ast_5,~~ \hat{F}_6=\hat{\bar{F}}^\ast_4.$$
\end{thm}

%

\begin{thm}\label{p1_thm16}
Let $C$ be a cyclic code of length $n$ over $\mathcal{R}$ with $C=\left\langle F \right\rangle_{R} \oplus U\left\langle F \right\rangle_{R} \oplus 2\left\langle F \right\rangle_{R}\oplus  2U\left\langle F \right\rangle_{R}\oplus (2+Um_f) \left\langle F \right\rangle_{R},$  where $m_{f}$ is an unit in $\dfrac{\mathbb{F}_4[x]}{\left\langle f \right\rangle}$. Then $C$ is the Hermitian self-dual code if and only if  $$F=\bar{F}^\ast,~~ G=\bar{G}^\ast,~~ H=\bar{H}^\ast,~~ Q=\bar{Q}^\ast,~~ T=\bar{T}^\ast.$$
\end{thm}

\begin{example}
The factorization of $x^7-1$ is $(x-1)(x^3+x+1)(x^3+x^2+1)$ over $\mathbb{F}_4$. Let $f_1=(x-1)$, $f_2=(x^3+x+1)$ and $f_3=(x^3+x^2+1)$, then $f_1=f^\ast_1$, $f_2=f^\ast_3$ and $f_3=f^\ast_2$. The following cyclic codes of length $7$ over $\mathcal{R}$ are self-dual (Euclidean) codes and their Gray image $\Phi(C)$ has parameters $[28,14,4]$ over $\mathbb{F}_4$. \\
$$\langle f_1f_2, \quad rf_2f_3 \rangle, \qquad \langle f_1f_3, \quad  rf_2f_3 \rangle, \qquad \mbox{where}~ r\in\{U,2,2+U\},$$
$$\langle 2Uf_1f_3, \quad 2f_1f_2, \quad Uf_1f_3, \quad sf_2f_3 \rangle, \qquad \mbox{where}~ s\in\{2,2+U\},$$
$$\langle 2Uf_1f_2,  \quad 2f_1f_3, \quad Uf_1f_3,  \quad tf_2f_3 \rangle, \qquad \mbox{where}~ t\in\{2,2+U\}.$$
\end{example}
 
\begin{example} 
The factorization of $x^5-1$ is $(x-1)(x^2+wx+1)(x^2+w^2x+1)$ over $\mathbb{F}_4$. Let $f_1=(x-1)$, $f_2=(x^2+wx+1)$ and $f_3=(x^2+w^2x+1)$, then $f_1=\bar{f}^\ast_1$, $f_2=\bar{f}^\ast_3$ and $f_3=\bar{f}^\ast_2$. The following cyclic codes of length $5$ over $\mathcal{R}$ are self-dual (Hermitian) codes and their Gray image $\Phi(C)$ has parameters $[20,10,4]$ over $\mathbb{F}_4$. \\
$$\langle f_1f_2, \quad rf_2f_3 \rangle, \qquad \langle f_1f_3, \quad  rf_2f_3 \rangle, \qquad \mbox{where}~ r\in\{U,2,2+U\},$$
$$\langle 2Uf_1f_3, \quad 2f_1f_2, \quad Uf_1f_3, \quad sf_2f_3 \rangle, \qquad \mbox{where}~ s\in\{2,2+U\},$$
$$\langle 2Uf_1f_2,  \quad 2f_1f_3, \quad Uf_1f_3,  \quad tf_2f_3 \rangle, \qquad \mbox{where}~ t\in\{2,2+U\}.$$
\end{example}

\section{Conclusion}
In 2013, Alahmadi et al. developed cyclic codes over finite matrix ring $M_2(\mathbb{F}_2)$ and their duals as right ideals in terms of two generators. Also cyclic codes over $M_2(\mathbb{F}_2)$ were made the existence of infinitely many nontrivial cyclic codes for Euclidean product. All this was derived of odd length code. In this paper, we constructed the structure of $M_2(\mathbb{Z}_4)$ and developed cyclic dual codes and cyclic self-dual codes over $M_2(\mathbb{Z}_4)$ which is even length codes over $M_2(\mathbb{F}_2)$. In \cite{Alahmadi2013} \cite{Luo2017}, it is not possible to construct negacyclic code, since the characteristic of structure of the ring is $2$. But in our construction one can form negacyclic code.  Welcome to the reader to construct even length codes over $M_2(\mathbb{Z}_4)$. Another useful direction for further study would be to consider LCD codes over $M_2(\mathbb{Z}_4)$.
\section*{Acknowledgements}
The author Sanjit Bhowmick is thankful to MHRD for financial support.




\end{document}